\documentclass[letterpaper, 10 pt, conference]{ieeeconf}
\IEEEoverridecommandlockouts                              

\usepackage[utf8]{inputenc} 
\usepackage[T1]{fontenc}    
\usepackage{hyperref}       
\usepackage{url}            
\usepackage{booktabs}       
\usepackage{amsfonts}       
\usepackage{nicefrac}       
\usepackage{microtype}      
\usepackage{lipsum}
\usepackage{graphicx}

\usepackage{xcolor}

\usepackage{amsmath}
\usepackage{amssymb}

\usepackage{algorithm}
\usepackage{algpseudocode}

\usepackage{mathtools} 

\usepackage{booktabs}

\usepackage{acronym}
\acrodef{OLS}{ordinary least squares estimator}
\acrodef{MLE}{maximum likelihood estimator}
\acrodef{SNR}{signal-to-noise ratio}
\acrodef{LTI}{linear time-invariant}

\newtheorem{theorem}{Theorem}
\newtheorem{corollary}{Corollary}

\newtheorem{remark}{Remark}
\newtheorem{assumption}{Assumption}
\newtheorem{prop}{Proposition}
\newtheorem{definition}{Definition}

\DeclareMathOperator*{\argmin}{arg\,min}

\DeclareMathOperator*{\diag}{diag}
\DeclareMathOperator*{\Tr}{Tr}

\newcommand{\R}{\mathbb{R}}
\newcommand{\Sbb}{\mathbb{S}}
\newcommand{\Prob}{\mathbb{P}}
\newcommand{\Expect}{\mathbb{E}}
\newcommand{\N}{\mathcal{N}}
\newcommand{\SSet}{\mathcal{S}}

\newcommand{\E}{\mathcal{E}}

\newcommand{\bigO}{\mathcal{O}}
\newcommand{\simiid}{\stackrel{\text{i.i.d.}}{\sim}}
\newcommand{\KL}[2]{\mathrm{KL}({#1} \Vert {#2})}
\newcommand{\kl}[2]{\mathrm{kl}({#1} \Vert {#2})}

\newcommand{\change}[1]{\textcolor{black}{#1}}

\title{Sample Complexity Bounds for Linear System Identification from a Finite Set}
\author{Nicolas Chatzikiriakos and Andrea Iannelli
	\thanks{This work is funded by Deutsche Forschungsgemeinschaft (DFG, German Research Foundation) under Germany's Excellence Strategy - EXC 2075 – 390740016. We acknowledge the support by the Stuttgart Center for Simulation Science (SimTech).}
	\thanks{The authors are with the University of Stuttgart, Institute for Systems
		Theory and Automatic Control, 70550 Stuttgart,
		Germany. (e-mail: \{\mbox{nicolas.chatzikiriakos}, {andrea.iannelli}\}@ist.uni-stuttgart.de)}%
}

\begin{document}

\maketitle

\begin{abstract}
    This paper considers a finite sample perspective on the problem of identifying an LTI system from a finite set of possible systems using trajectory data. 
    To this end, we use the maximum likelihood estimator to identify the true system and provide an upper bound for its sample complexity. 
    Crucially, the derived bound does not rely on a potentially restrictive stability assumption.
    Additionally, we leverage tools from information theory to provide a lower bound to the sample complexity that holds independently of the used estimator.     
    The derived sample complexity bounds are analyzed analytically and numerically. 
\end{abstract}


\section{Introduction}
In this work, we consider the problem of identifying an \ac{LTI} system from a finite set of  system models using data from a finite and noisy trajectory. 
To this end, we use tools from statistical learning theory and information theory to derive  high probability upper and lower bounds of the sample complexity of identifying the true system using the \ac{MLE}. 

The problem of identifying the true system from a given finite set of systems naturally appears in many applications~\cite{Choirat2012}. 
In particular, switched systems are prevalent across different domains such as mechanical systems or the automotive industry~\cite{Liberzon1999}. 
Identifying the active mode of a switched system is also a relevant problem for control, as shown, e.g., by several works in the adaptive control literature~\cite{Narendra1997}. Recently, renewed interest has emerged in this literature for considering uncertainty in the form of a finite set of models~\cite{Rantzer2021}.
An additional application domain motivating our interest are ecological and evolutionary models 
where often the correct hypothesis out of a finite hypothesis class needs to be determined from observations~\cite{Johnson2004}.
Choosing an element from a finite set by leveraging information coming from measured data is also paradigmatic of many decision-making problems, e.g., in the bandits literature~\cite{Lattimore2020}. 
A notable example is best-arm identification, 
a pure exploration problem that has been successfully applied, e.\,g., for clinical trials~\cite{Villar2015}. In the fixed confidence setting~\cite{Jamieson2014}, the goal is to identify the best arm among a selection of arms with a desired confidence. 
In the dynamical systems setting, this translates into the high-probability identification of the true system from a finite set of systems.
A key difference, to the setup considered in this work is that best-arm identification assumes static arms and uses the control input to achieve the goal as fast as possible 
In this work, we assume Gaussian control inputs 
however the system dynamic introduces correlation we need to address.
%
%
\paragraph*{Related Works}
While in the asymptotic regime, the statistic analysis of system identification has a long history~\cite{ljung1999system}, recently novel tools in statistical learning theory~\cite{Abbasi2011Improved, wainwright2019high} sparked an increasing interest in non-asymptotic results, analyzing error bounds of estimated models, especially the \ac{OLS}, from a finite-sample perspective. 
While first works~\cite{dean2020sample,matni2019tutorial} relied on i.i.d data 
the now predominant part of the literature analyses the case of trajectory data.
Hereby,~\cite{simchowitz2018learning} was the first work to provide a finite sample analysis for fully observed marginally stable systems. 
The analysis was extended to unstable systems~\cite{sarkar2019near} uncovering a statistical inconsistency of the \ac{OLS} under certain conditions. 
Results providing lower bounds on the sample complexity~\cite{jedra2019sample} further enabled a fundamental analysis and understanding of difficulties in identifying linear systems from trajectory data~\cite{tsiamis2021linear}. 
The finite sample analysis of the identification of switched systems using \ac{OLS} has been considered in~\cite{Sarkar2019}, with known switching sequences and unknown system matrices. 
For a comprehensive overview of non-asymptotic system identification, we refer to~\cite{Tsiamis2023}.
\paragraph*{Contribution}
To the best of our knowledge, we provide the first sample complexity analysis for the problem of identifying an unknown \ac{LTI} system from a finite set of models. 
To this end, we derive an instance specific sample complexity upper bound for the \ac{MLE}.
Hereby, the tools used to analyze the unconstrained \ac{OLS} when identifying an unknown system from a continuous set can not be applied, since they rely on the closed-form solution of the \ac*{OLS}.
In contrast, our proof relies on showing high-probability concentration of the cost of the \ac*{MLE} for the true system and anti-concentration for all other systems. 
A further advantage of the proposed approach is that, unlike most finite sample results for the continuous set case, the results in this work do not rely on any stability assumptions.  
In addition, we provide an instance specific sample complexity lower bound for $\delta$-stable algorithms. 
An analytical and numerical analysis of the bounds shows which factors influence the hardness of identification.
%
\paragraph*{Notation}
The unit sphere in $\R^n$ is denoted by $\mathbb{S}^{n-1}$.
Given a matrix $M$ we denote 
its Frobenius norm by $\Vert M \Vert_\mathrm{F}$. 
Given a vector $x\in \R^{n}$ and a matrix $P\succ 0$ we define $\Vert x\Vert_{P}^2 = x^\top P x$. In the special case of $P=I$ we omit the subscript. 
We denote matrix blocks that can be inferred from symmetry by $\star$, i.\,e., we write $\Lambda^\top \Sigma \Lambda =[\star] \Sigma \Lambda$.
We denote the canonical basis in $\R^n$ with $e_1,\dots, e_n$.
Given some $z\in \R$ we write $\lfloor z \rfloor$ to denote the floor-function.  
We use the shorthand $\Prob_\theta[\cdot]$ ($\Expect_\theta[\cdot]$) when referring to the probability (expectation) given that the system generating the data is given by $\theta$.
%
%
\section{Preliminaries}
\subsection{Problem setup}\label{sec:ProblemSetup}
Consider the  linear time-invariant discrete-time system
\begin{equation}
    x_{t+1} = A_* x_t + B_* u_t + w_t, \label{eq:TrueSysEvo}
\end{equation}
where $x_t\in \R^{n_x}$ is the state of the system, $u_t\in\R^{n_u}$ is the control input and $w_t \in \R^{n_x}$ is unknown process noise.  
We seek to identify the unknown system matrices $\theta_* = (A_*, B_*)$ from a single trajectory $\{x_t\}_{t=1}^T$, $\{u_t\}_{t=1}^{T-1}$ of length $T$. 
We assume the data is collected as follows.
\begin{assumption}\label{ass:Gaussian}
The process noise and control input are i.\,i.\,d. zero-mean Gaussian with known covariance matrices  $\Sigma_w, \Sigma_u \succ 0$, i.\,e., $ w_t\simiid \mathcal{N}(0, \Sigma_w)$  and $u_t\simiid \mathcal{N}(0, \Sigma_u)$.
\end{assumption}
\change{While Assumption~\ref{ass:Gaussian} is standard in non-asymptotic identification literature \cite{simchowitz2018learning,Tsiamis2023}, extensions to other noise classes are an interesting topic for future work.} 
Further, we assume to know a finite set of possible systems which contains the true system, i.\,e., ${\theta_* \in \SSet := \{\theta_0, \dots, \theta_N\}}$, where $\theta_i = (A_i, B_i)$.
We assume $(A_*, B_*) = (A_0, B_0)$ to simplify the notation.
\subsection{Maximum likelihood estimation}
To identify the true system $(A_*, B_*)$ from data we resort to the \ac{MLE}, which is asymptotically optimal and efficient \cite{Myung2003}. 
Here, we analyze its non-asymptotic behavior, by means of statistical learning theory tools. 
To this end, observe that the probability of collecting the data $d_t \coloneqq \{x_{t+1}, x_t, u_t\}$ from system $i$ is given by 
\begin{equation*}
    \Prob_{\theta_i}(d_t) = \frac{1}{\sqrt{(2\pi)^{n_x}\vert \Sigma_w \vert }} e^{-\frac12[\star] \Sigma_w^{-1}(x_{t+1} - A_i x_t - B_i u_t)}.
\end{equation*}
Based on this observation, we define the cost 
\begin{equation}\label{eq:cost}
    \ell_{\theta_i} (x_t, u_t) = \Vert x_{t+1} - A_i x_t - B_i u_t \Vert_{\Sigma_w^{-1}}^2,
\end{equation}
which is proportional to the negative log-likelihood of observing the data $\{x_{t+1}, x_t, u_t\}$ from system $i$.
Thus, the \ac{MLE} minimizes the empirical risk ${\hat L (\theta) = \frac 1T \sum_{t=1}^T \ell_{\theta}(x_t, u_t)}$, i.\,e., 
\begin{equation} \label{eq:MLE}
    \hat \theta_T \in \argmin_{\theta \in \SSet} \hat L (\theta).
\end{equation}
For our analysis we adopt the notion of sample complexity (see, e.\,g., \cite{Tsiamis2023})
and aim at deriving an instance specific sample complexity upper bound for the \ac{MLE}~\eqref{eq:MLE}. 
That is, given a chosen failure probability $\delta \in (0,1)$, we provide guarantees of the form 
\begin{equation} \label{eq:sampleCompUpper}
    \Prob[\hat \theta_T = \theta_*] \ge 1-\delta, \quad \text{if } T \ge T_\mathrm{ub},
\end{equation}
where $T_\mathrm{ub} = T_\mathrm{ub}(\delta, \mathcal{A}, \mathcal{S}, \theta_*)$ is an upper bound of the problem specific sample complexity of an analyzed estimation algorithm $\mathcal{A}$. 
Additionally, we provide an instance-specific sample complexity lower bound that holds independently of the used estimator, i.\,e., we show that any 
\begin{equation}\label{eq:sampleCompLower}
    \Prob\left[\hat \theta_T = \theta_*\right] \le 1-\delta \quad \text{if } T \le T_\mathrm{lb},
\end{equation}
where $T_\mathrm{lb} = T_\mathrm{lb}(\delta, \mathcal{S}, \theta_*)$ is a lower bound of the problem specific sample complexity for any estimation algorithm yielding the estimate $\hat \theta_T$. 

\section{Finite sample identification}
%
%
\subsection{A sample complexity upper bound}
Plugging the dynamics~\eqref{eq:TrueSysEvo} into the cost~\eqref{eq:cost} yields
\begin{align*}
    \ell_{\theta_i}(x_t, u_t) 
    &= \Vert w_t + \Delta A_i x_t + \Delta B_i u_t \Vert_{\Sigma_w^{-1}}^2,
\end{align*}
where we defined $\Delta A_i \coloneqq A_* - A_i$ and $\Delta B_i \coloneqq B_* - B_i$.
Further, we define 
\begin{equation}\label{eq:defz}
    z_t^i \coloneqq  \Sigma_w^{-\nicefrac{1}{2}}(w_t + \Delta A_i x_t + \Delta B_i u_t),
\end{equation}
where $\Sigma_w^{-1} = \left(\Sigma_w^{-\nicefrac{1}{2}}\right)^\top \Sigma_w^{-\nicefrac{1}{2}}$. With this, the empirical risk reads 
${
    \hat L (\theta_i) 
    =   \frac{1}{T}\sum_{t=1}^{T}  \Vert z_t^i\Vert^2 
}$.
Since the sum of Gaussian random variables is Gaussian we have that for all $t \in [1, T]$
\begin{equation}
    \label{eq:Distributions}
    x_t \sim \mathcal{N}\Bigg(A^t x_0, \underbrace{\sum_{\tau=0}^t [\star] \Sigma_u ( A_*^\tau B_*)^\top +  [\star] \Sigma_w {A_*^\tau}^\top}_{:=\Sigma_{x_t}} \Bigg).
\end{equation}
With this, we see that $z_t^i$ is Gaussian distributed as well, i.\,e., 
\begin{equation} \label{eq:distZ}
    z_t^i \sim \mathcal{N}\Big(\Sigma_w^{-\nicefrac{1}{2}}\Delta A_i A^t x_0, \Sigma_{z_t}^i\Big),
\end{equation}
where 
\begin{equation} \label{eq:varZ}
    \Sigma_{z_t}^i \coloneqq I + \Sigma_w^{-\nicefrac{1}{2}} \left([\star] \Sigma_u \Delta B_i^\top + [\star] \Sigma_{x_t} \Delta A_i^\top\right) {\Sigma_w^{-\nicefrac{1}{2}}}^\top.
\end{equation}
It is important to note that the sequence $(z_t^i)_{t\ge 1}$ is highly correlated due to the underlying LTI system generating the data. 
To handle the correlation, we leverage the block martingale small-ball condition introduced in \cite{simchowitz2018learning}. 
\begin{definition}[Block Martingale Small-Ball \cite{simchowitz2018learning}]
    Let $(\zeta_t)_{t \ge 1}$ be a $\{\mathcal{F}_t\}_{t\ge 1}$-adapted random process taking values in $\R$.
    We say $(\zeta_t)_{t \ge 1}$ satisfies the $(k, \nu, p)$-block martingale small-ball (BMSB) condition if, for any $j \ge 0$
    \begin{equation*}
        \frac1k \sum_{i=1}^k \Prob[\vert \zeta_{j+i}\vert \ge \nu\vert \mathcal{F}_j] \ge p \quad \text{a.s.}
    \end{equation*}
    Given a process $(z_t)_{t\ge 1}$ taking values in $\R^{n_z}$, we say that it satisfies the $(k, \Gamma_{\mathrm{sb}}, p)$-BMSB condition for $\Gamma_\mathrm{sb} \succ 0$ if, for any fixed $v \in \mathbb{S}^{n_x-1}$ the process $\zeta_t = \langle v, z_t \rangle$ satisfies $(k, \sqrt{v^\top \Gamma_\mathrm{sb}v}, p)$-BMSB.
\end{definition}
The BMSB condition establishes 
a level of 
anti-concentration along a sequence. 
We can show that the sequence $(z^i_t)_{t= 1}^T$ satisfies the BMSB condition. 
\begin{prop}\label{prop:BMSB}
    Let Assumption~\ref{ass:Gaussian} hold, let $z_t^i$ be defined according to~\eqref{eq:defz} and define 
    \begin{equation*}
        \mathcal{F}_t \coloneqq (w_0, \dots, w_t, x_0, \dots, x_t, u_0,\dots, u_t).
    \end{equation*} 
    Then the $\{\mathcal{F}_t\}_{t=1}^T$-adapted random process $(z_t^i)_{t=1}^T$ satisfies the $\left( k, \Sigma_{z_{k/2}}^i, \nicefrac{3}{20}\right)$-BMSB condition for all ${k\in [1,T]}$.
\end{prop}
The proof of Proposition~\ref{prop:BMSB} builds on arguments used in \cite{simchowitz2018learning} and can be found in Appendix~\ref{sec:ProofBMSB}. 
Before we state the main theorem of this section we define 
\begin{subequations}
    \begin{align}
        \Delta\Lambda_i^u(t) &\coloneqq \Delta A_i\sum_{s=0}^{t} A_*^s B_*\Sigma_u^{\nicefrac{1}{2}} \label{eq:GramLikeU}\\ 
        \Delta \Lambda_i^w(t) &\coloneqq \Delta A_i\sum_{s=0}^{t} A_*^s\Sigma_w^{\nicefrac{1}{2}}\label{eq:GramLikeW}
    \end{align}   
    \label{eq:GramLike} 
\end{subequations}
describing the excitations due to the control input~\eqref{eq:GramLikeU} and noise~\eqref{eq:GramLikeW} projected on the differences between $A_*$ and $A_i$.  
\begin{theorem}\label{th:upperBound}
    Let $\{x_t\}_{t=1}^T$, $\{u_t\}_{t=1}^{T}$ be data collected from system~\eqref{eq:TrueSysEvo} according to Assumption~\ref{ass:Gaussian}. Fix a failure probability $\delta  \in (0, 1)$.  
    \begin{subequations}
    If there exists $k\in [1,T]$ such that 
    \begin{align} \label{eq:ThMainBurnIn}
        \lfloor\nicefrac{T}{k}\rfloor &\ge \nicefrac{320}{3} \log(\nicefrac{2n_x N}{\delta}) 
    \end{align}
        and for all $i \in [1, N]$ 
    \begin{align}
            \sqrt{n_x} + n_x  &\le \frac{9k\lfloor \nicefrac{T}{k}\rfloor }{3200T} \Bigg(\Vert\Sigma_w^{-\nicefrac{1}{2}}\Delta \Lambda_i^w(\nicefrac{k}{2}) \Vert_\mathrm{F}^2 + n_x \label{eq:ThMainCond} \\ & \quad +  \Vert \Sigma_w^{-\nicefrac{1}{2}} \Delta B_i \Sigma_u^{\nicefrac{1}{2}}\Vert_\mathrm{F}^2  + \Vert\Sigma_w^{-\nicefrac{1}{2}}\Delta\Lambda_i^u(\nicefrac{k}{2})  \Vert_{\mathrm{F}}^2 \Bigg), \notag
    \end{align}
    then the \ac{MLE}~\eqref{eq:MLE} yields the true system $\theta_*$ with probability at least $1-\delta$, i.\,e.,  $\Prob[\hat \theta_T = \theta_*] \ge 1-\delta$. 
    \end{subequations} 
\end{theorem}
\begin{proof} The proof is articulated in three steps. First, we use the BMSB condition to show anti-concentration of $\frac1T \sum_{t=1}^T \Vert z_t^i \Vert$. 
Secondly, we show that $\frac1T \sum_{t=1}^T \Vert w_t \Vert$ concentrates. 
Finally, we combine the previous two results to obtain that $\hat L(\theta_*) < \hat L(\theta_i) \quad \forall \theta_i\neq\theta_*$
with high probability.
%
%
\paragraph{Lower bounding the empirical risk of $\theta  \neq \theta_*$}
Recall that Proposition~\ref{prop:BMSB} shows that the process $(z_t^i )_{t=1}^T$ satisfies the $\left( k, \Sigma_{z_{k/2}}^i, \nicefrac{3}{20}\right)$-BMSB condition. 
Note that $e_\ell \in \mathbb{S}^{n_x-1} \, \forall l \in [1, n_x]$. Thus, for some fixed $\ell\in[1, n_x]$, it follows 
from the BMSB condition 
that 
\begin{equation*}
    \frac1k \sum_{t=1}^k \Prob \left[\vert [z^i_{s+t}]_\ell \vert \ge  \sqrt{{e_\ell}^\top \Sigma_{z_{k/2}}^i e_\ell}\right] \ge \frac{3}{20},
\end{equation*}
where $[\cdot]_\ell$ extracts the $\ell$-th element from a vector. 
By applying Corollary\,\ref{co:SimchoTight} (Appendix~\ref{app:SimchoTight}), which is a tighter version of \cite[Proposition 2.5]{simchowitz2018learning}, we obtain the  anti-concentration result
\begin{equation}
        \Prob\Bigg[\sum_{t=1}^{T}{[z_t^i]}^2_\ell \le \frac{9k\lfloor \nicefrac{T}{k}\rfloor }{3200}{e_\ell}^\top \Sigma_{z_{k/2}}^i e_\ell \Bigg] \le e^{-\frac{3}{320}\lfloor \nicefrac{T}{k}\rfloor}.
    \label{eq:antiConcentr}
\end{equation}
In the following we impose $\exp(-\nicefrac{3}{320}\lfloor \nicefrac{T}{k}\rfloor) \le \frac{\delta}{2n_x N }$, which requires the burn-in time condition
\begin{equation}\label{eq:burnIn1}
   \lfloor \nicefrac{T}{k}\rfloor \ge \nicefrac{320}{3} \log(\nicefrac{2n_x N}{\delta}),
\end{equation}
which is~\eqref{eq:ThMainBurnIn}. We consider the events
\begin{equation*}
    \mathcal{E}_{\ell}:= \sum_{t=1}^{T}{[z_t^i]}^2_\ell \ge \nicefrac{9k}{320} \lfloor \nicefrac{T}{k}\rfloor {e_\ell}^\top \Sigma_{z_{k/2}}^i e_\ell, 
\end{equation*}
as well as their union 
$\mathcal{E} := \bigcup\limits_{\ell=1}^{n_x} {\mathcal{E}}_{\ell}$.
It follows that 
\begin{equation*}
    \mathcal{E} \implies \sum_{\ell=1}^{n_x} \sum_{t=1}^{T}{[z_t^i]}^2_\ell \ge \sum_{\ell=1}^{n_x}
    \nicefrac{9k}{3200} \lfloor \nicefrac{T}{k}\rfloor {e_\ell}^\top \Sigma_{z_{k/2}}^i e_\ell.
\end{equation*}
Because of~\eqref{eq:burnIn1} we have $\Prob[\mathcal{E}_{\ell}]\ge 1- \frac{\delta}{2n_x N }$. Hence, using union bound arguments, we obtain
\begin{alignat*}{1}
    \Prob\Bigg[\sum_{t=1}^{T}& \Vert z_t^i\Vert^2                                                                                        \ge \sum_{\ell=1}^{n_x}
     \nicefrac{9k}{3200} \lfloor \nicefrac{T}{k}\rfloor {e_\ell}^\top \Sigma_{z_{k/2}}^i e_\ell\Bigg]                                                                                  \\ &= \Prob\left[\sum_{\ell=1}^{n_x} \sum_{t=1}^{T}{[z_t^i]}^2_\ell \ge \sum_{\ell=1}^{n_x}
     \nicefrac{9k}{3200}  \lfloor \nicefrac{T}{k}\rfloor {e_\ell}^\top \Sigma_{z_{k/2}}^i e_\ell\right] \\                         & \ge \Prob\left[\mathcal{E}\right] \ge \prod_{l=1}^{n_x}\left(1-\frac{\delta}{2n_x N }\right) \ge 1- \frac{\delta}{2N}.         
\end{alignat*}
Finally, observe $\sum_{\ell=1}^{n_x} {e_\ell}^\top \Sigma_{z_{k/2}}^i e_\ell = \Tr\left(\Sigma_{z_{k/2}}^i\right)$
to obtain
\begin{equation}\label{eq:resultFalseSys}
        \Prob\Bigg[\sum_{t=1}^{T}\Vert z_t^i\Vert^2                                                                             \ge  \nicefrac{9k}{3200} \lfloor \nicefrac{T}{k}\rfloor \Tr\left(\Sigma_{z_{k/2}}^i\right)\Bigg] \ge 1- \frac{\delta}{2N}.
\end{equation}
%
%
\paragraph{Upper bounding the empirical risk of $\theta_0$}
Note that 
$\Delta A_0 = \Delta B_0 = 0$. 
Set $\zeta_t^\top = w_t^\top \Sigma^{-\frac12}_w \simiid \mathcal{N}(0, I)$ to obtain
\begin{align*}
    \hat L (\theta_0) = \frac 1T \sum_{t=1}^T \Vert w_t \Vert_{\Sigma^{-1}_w} =  \frac1T \sum_{t=1}^T \zeta_t^\top \zeta_t  = \frac{1}{T} \sum_{t=1}^{n_xT}\xi_t^2,
\end{align*}
where $\xi_k \stackrel{\text{i.i.d.}}{\sim} \mathcal{N}(0, 1)$. 
Clearly, $\xi_k^2$ is sub-exponential with parameters $(4, 4)$. 
Since the sum of sub-exponential random variables $X_i$ with parameters $( \nu_i^2, \alpha_i)$  is sub-exponential with parameters $(\sum_i \nu_i^2, \max_i(\alpha_i))$ \cite{wainwright2019high}, we have ${\sum_{t=1}^{T \cdot n_x} \xi_t^2 \sim \mathrm{subExpo}\left(4Tn_x, 4\right)}$.
Using \cite[Proposition 2.9]{wainwright2019high} with $t=\sqrt{n_x} T$  after minor reformulations we obtain 
\begin{equation}
    \Prob\left[\frac{1}{T} \sum_{k=1}^{n_xT} \xi_k^2 \ge \sqrt{n_x} + n_x\right] \le \exp(-\nicefrac{T}{8}) \le \frac{\delta}{2},
    \label{eq:BoundTrueSysFinal}
\end{equation}
where the last inequality is satisfied if  
\begin{equation} \label{eq:burnIn2} 
    T \ge 8 \log(\nicefrac{2}{\delta}).
\end{equation}
%
%
\paragraph{Leveraging concentration and anti-concentration} 
First, note that the burn-in time condition~\eqref{eq:burnIn1} implies the burn-in time condition~\eqref{eq:burnIn2}. Hence, from~\eqref{eq:BoundTrueSysFinal} we have 
\begin{equation*}
              \Prob[\mathcal{E}_{\theta_0}]\coloneqq\Prob\left[\hat L(\theta_0) \le  \sqrt{n_x} + n_x \right] \ge 1 - \frac{\delta}{2}.
\end{equation*}
Because of~\eqref{eq:resultFalseSys} for each $\theta_i$, $i\in [1,N]$ it holds that
\begin{align*}
              \Prob[\mathcal{E}_{\theta_i}] &\coloneqq \Prob\left[\hat L(\theta_i) \ge \frac{9k\lfloor \nicefrac{T}{k}\rfloor}{3200 T} \Tr\left(\Sigma_{z_{k/2}}^i \right) \right] \ge 1- \frac{\delta}{2N}.
\end{align*}
If $\sqrt{n_x} + n_x  < \frac{9k\lfloor \nicefrac{T}{k}\rfloor}{3200 T} \Tr( \Sigma_{z_{k/2}}^i)$ holds  $\forall i \in [1, N]$ we have 
\begin{equation*}\label{eq:overallProb}
        \Prob\left[\hat \theta = \theta_*\right] \ge \Prob\left[\bigcup\limits_{i=0}^{N} \mathcal{E}_{\theta_i}\right] = \left(1-\tfrac{\delta}{2}\right)\prod_{i=1}^N \left(1-\tfrac{\delta}{2 N }\right) \ge 1-\delta,
\end{equation*}
where the second inequality follows from union-bound arguments.
From here we can conclude the proof, by using linearity of the trace, as well as $\Tr(AA^\top) = \Vert A\Vert_\mathrm{F}^2$ to obtain 
\begin{align*}
    \Tr\left( \Sigma_{z_{k/2}}^i \right) &= \Vert\Sigma_w^{-\nicefrac{1}{2}}\Delta \Lambda_i^w(\nicefrac{k}{2}) \Vert_\mathrm{F}^2 + n_x  \\ & \quad +  \Vert \Sigma_w^{-\nicefrac{1}{2}} \Delta B_i \Sigma_u^{\nicefrac{1}{2}}\Vert_\mathrm{F}^2  + \Vert\Sigma_w^{-\nicefrac{1}{2}}\Delta\Lambda_i^u(\nicefrac{k}{2})  \Vert_{\mathrm{F}}^2.
\end{align*}
\end{proof}
\begin{remark}\label{rem:stability}
    \change{Theorem~\ref{th:upperBound} imposes no stability assumptions on the system~\eqref{eq:TrueSysEvo}. This is in contrast with non-asymptotic results for the unconstrained \ac{OLS}, where a statistical inconsistency has been shown for certain classes of unstable systems~\cite{sarkar2019near}.}
\end{remark}
Note that any $T$ satisfying the conditions in Theorem~\ref{th:upperBound} is an upper bound to the sample complexity as defined in~\eqref{eq:sampleCompUpper}. 
\change{Further, the cardinality $N+1$ of the set $\mathcal{S}$ enters Theorem~\ref{th:upperBound} in \eqref{eq:ThMainBurnIn} and \eqref{eq:ThMainCond}. Whereas the former shows the dependence $\bigO(\log N)$, the influence of $N$ on \eqref{eq:ThMainCond} is more subtle and depends on the particular systems being added as $N$ increases.}
\change{Note also, that Theorem~\ref{th:upperBound} depends on the true systems matrices, which are unknown in practice. While data-dependent results might prove more useful from a practical perspective, the value of Theorem~\ref{th:upperBound} lies in understanding the fundamental difficulty of the learning problem.}
To investigate the conservatism of Theorem~\ref{th:upperBound} analytically, we now derive finite-sample identification lower bounds. 
\subsection{A sample complexity lower bound}
In this section, we provide a sample complexity lower bound for the class of $\delta$-stable estimation algorithms, which we define as follows. 
\begin{definition}[$\delta$-stable algorithms]
    Consider the setup described in Section~\ref{sec:ProblemSetup}. An algorithm is called $\delta$-stable, if for all $\delta\in (0,1)$, and any $\theta_*$ and $\SSet$ there exists a finite time $\bar T$ s.t. for all $t\ge \bar T$ we have $\Prob_{\theta_*}(\hat \theta_t = \theta_*)$. 
\end{definition}
The notion of $\delta$-stable algorithms excludes algorithms that yield the same estimate independently of the data observed and is inspired by $(\varepsilon, \delta)$-locally-stable algorithms in \cite{jedra2019sample}. 
\begin{theorem}\label{th:lowerBound}
     Let Assumption~\ref*{ass:Gaussian} hold. 
    Then, for any $\delta$-stable algorithm, all  $\delta \in (0,1)$ and all $i \in [1, N]$ it holds that 
    \begin{align} \label{eq:lowerBound}
        \bar T \Big\Vert \Sigma_w^{-\frac12} \Delta B_i \Sigma_u^{\nicefrac12}\Big\Vert _\mathrm{F}^2  &+  \sum_{s=0}^{\bar T -1}  \left\Vert \Sigma_w^{-\nicefrac12} \Delta\Lambda_{i}^u(s)\right\Vert _\mathrm{F}^2  \\ &\quad + \left\Vert \Sigma_w^{-\nicefrac12} \Delta \Lambda_{i}^w(s)\right\Vert _\mathrm{F}^2   \ge 2 \log\left(\tfrac{1}{2.4 \delta}\right). \notag
    \end{align}
\end{theorem}
The proof of Theorem~\ref{th:lowerBound} is inspired by~\cite{jedra2019sample} and can be found in Appendix~\ref{app:proofLower}.
Importantly, Theorem~\ref{th:lowerBound} can be analyzed to see which factors contribute to identifying the true system. 

%
%
\subsection{Analysis of the sample complexity bounds}\label{sec:comp}
Having established a sample complexity upper bound for the \ac{MLE}~\eqref{eq:MLE} as well as an estimation algorithm independent lower bound, we can now analyze and compare Theorems\,\ref{th:upperBound} and~\ref{th:lowerBound}. 
To this end, we will focus on three key factors that influence the identification of the true system.
%
%
\paragraph{Excitation and noise level} For simplicity consider the case where $\Sigma_w = \sigma_w^2 I$ and $\Sigma_u = \sigma_u^2 I$. In this case, according to both the upper bound (Theorem~\ref{th:upperBound}) and the lower bound (Theorem~\ref{th:lowerBound}) increasing the ratio $\nicefrac{\sigma_u}{\sigma_w}$ allows for either a decrease in the number of samples $T$ or in the failure probability $\delta$.
%
%
\paragraph{Excitation directions} 
Observe that condition~\eqref{eq:ThMainCond} of the upper bound and condition~\eqref{eq:lowerBound} of the lower bound qualitatively both depend on the same three terms and can be interpreted as a \ac{SNR} condition. 
To this end, recall that $\Delta \Lambda_i^B(t)$, $\Delta B_i \Sigma_u^\frac12$ and $\Delta \Lambda_i^I(t)$ can be interpreted as the excitation of the system projected to the difference between systems $\theta_*$ and $\theta_i$. 
Weighting this measure of relevant excitation with the covariance of the noise yields an effective \ac{SNR}. 
Importantly, the directions of the excitation matter. 
That is, using the control input to excite the system where $\Delta A_i$ is large results in a smaller burn-in time or failure probability. 
Further, if the noise affects some states more than others this will also affect the lower and upper bounds.
%
%
\paragraph{Sample efficiency} 
Assume $\Sigma_w$ and $\Sigma_u$ are fixed. Then, decreasing the failure probability $\delta$ requires a larger $T$ to satisfy the lower bound~\eqref{eq:lowerBound}, i.\,e., the failure probability can be decreased when more samples are available.
Regarding the upper bound derived in Theorem~\ref{th:upperBound}, the burn-in time condition~\eqref{eq:ThMainBurnIn} shows a similar coupling between $T$ and $\delta$.  
On the other hand, the \ac{SNR}-condition~\eqref{eq:ThMainCond} does only exhibit a weak dependence on $T$ through the parameter $k$. An increase in $T$ allows for a larger $k$ which in turn enters~\eqref{eq:ThMainCond} through $\Delta \Lambda_i^B(\nicefrac{k}{2})$ and $\Delta \Lambda_i^I(\nicefrac{k}{2})$. 
The dependence of these quantities on $k$ depends heavily on the stability properties of the system. 
As expected, the more stable $A_*$, the harder it is to satisfy~\eqref{eq:ThMainCond}. 
To conclude, it has to be said that, even though the upper bound (Theorem~\ref{th:upperBound}) and lower bound (Theorem~\ref{th:lowerBound}) qualitatively depend on similar quantities, there still is a substantial gap between the two. 
This is largely due to the leading constants in condition~\eqref{eq:ThMainBurnIn} and~\eqref{eq:ThMainCond}, which appear due to the BMSB condition needed because of the correlation in the data. 

\section{Numerical example}\label{sec:numericalExample}
In the following, we investigate the results and observations of the previous sections using a numerical example.\footnote{The Python code for the numerical example can be accessed at: \url{https://github.com/col-tasas/2024-bounds-finite-set-ID}}
To do so we consider the set ${\mathcal{S} = \{(A_0,B), (A_1, B), (A_2, B)\}}$, with 
\begin{align*}
    A_i = \begin{bmatrix}
        a_{i} & 0.1 & 0 \\
        0 & 0.2 & 0 \\
        0 & 0 & b_{i}  
    \end{bmatrix}, 
    \quad
    B = \begin{bmatrix}
       0 & 0 \\
       1 & 0 \\
       0 & 1 
    \end{bmatrix},
\end{align*}
where $a_0=a_2=0.2$, $a_1= 0.1$, $b_0 = b_1=0.5$ and $b_2 = 0.6$. 
\change{We choose a small cardinality of $\SSet$ and low state space dimension since this setup suffices to make a number of key observations while maintaining clarity of exposition.}
Note first, that each $\theta_i \in \SSet$ has a weak coupling between $x_1$ and $u$ making it hard to excite the first mode of the system. This is a structure known to make identification hard using \ac{OLS}~\cite{tsiamis2021linear} and also plays a key role here.
To show the influence of the directions of excitation on the identification of the true system, we conduct three numerical experiments:
\begin{description}
    \item[Exp.\,1:]~$\Sigma_u = \diag(10, 0.1)$, $\Sigma_w = 0.1 I$ 
    \item[Exp.\,2:]~$\Sigma_u = \diag(0.1, 10)$, $\Sigma_w = 0.1 I$
    \item[Exp.\,3:]~$\Sigma_u = \diag(10, 0.1)$, $\Sigma_w = \diag(10, 0.1, 0.001)$.
\end{description}
Table~\ref{tab:Example} shows the percentage of the estimates (within 1000 trials) for varying $T$ for estimation using \ac{MLE} as presented in this work as well as using \ac{OLS} and projecting on the closest system in spectral norm. 
\begin{table*}[t!]
    \caption{Estimation percentages for different number of samples. Lower bound is satisfied with $\delta = 0.05$ for $T \ge 192$ ({Exp.~1}), $T\ge 400$ (Exp.~2) and $T \ge 404$~(Exp~3). Upper bound is not satisfied for any of the displayed experiments and sample sizes.}
    \label{tab:Example}
    \centering
    \footnotesize
    \begin{tabular}{@{}rccccccccccc@{}}
        \toprule 
        $T$ & \multicolumn{3}{c}{$\Prob_\mathrm{MLE}[\theta_0]$ ($\Prob_\mathrm{OLS}[\theta_0]$) in $\%$} & & \multicolumn{3}{c}{$\Prob_\mathrm{MLE}[\theta_1]$ ($\Prob_\mathrm{OLS}[\theta_1]$) in $\%$} && \multicolumn{3}{c}{$\Prob_\mathrm{MLE}[\theta_2]$ ($\Prob_\mathrm{OLS}[\theta_2]$) in $\%$}\\
        \cmidrule{2-4} \cmidrule{6-8} \cmidrule{10-12}
        & Exp 1 & Exp 2 & Exp 3 && Exp 1 & Exp 2 & Exp 3&&  Exp 1 & Exp 2 & Exp 3 \\
        \midrule
        250     & 80.2 (72.3)     & 77.7 (71.9)    & 78.3 (73.9) &&  10.9 (16.1)   & 22.3 (27.0)       & 21.7 (23.4) && 8.9 (11.6) & 0.0 (1.1) & 0.0 (2.7)      \\    
        500 & 92.8 (86.8) & 87.6 (84.2)  & 87.7 (84.7)    && 4.7 (7.9) & 12.4 (15.3 )& 12.3 (12.7) && 2.5 (5.3) & 0.0 (0.5) & 0.0 (2.6)\\
        750 & 96.6 (92.9) & 91.2 (87.8)       & 90.7 (88.4)        && 2.5 (3.8)& 8.8 (12.1)& 9.3 (9.7)&& 0.9 (3.3)& 0.0 (0.1) & 0.0 (1.9)\\
        1000 & 98.5 (96.9) & 95.1 (93.9)  & 94.8 (92.9)&& 1.3 (2.3) & 4.9 (6.1)& 5.2 (5.2) && 0.2 (0.8) &  0.0 (0.0) &0.0 (1.9)\\
        1250& 99.4 (99.0) & 98.1 (97.3)   & 95.4 (93.7) && 0.5 (0.7) & 1.9 (2.7)& 4.6 (4.7) && 0.1 (0.3) & 0.0 (0.0) &0.0 (1.6)\\
        \bottomrule
    \end{tabular} 
\end{table*}
The results in Table\,\ref{tab:Example} allow us to numerically show the observations made in Section~\ref{sec:comp}.
Firstly, the true positive rate increases as $T$ increases. 
Secondly, the different directions of excitation in the experiments play an important role. Clearly, Exp.~1 allows for the fastest identification of the true system, both in the numerical simulation and in the lower bound.
When considering Exp.~2 and Exp.~3 we can observe that $\theta_2$ can be ruled out very quickly because it differs from the true system in the $x_3$ directions which has a very high SNR (either through large excitation or low noise).   
Since the excitation of $x_1$ is very low in Exp.~2 and the noise affecting $x_1$ is very large in Exp.~3 distinguishing between $\theta_0$ and $\theta_1$ takes longer than in Exp.~1. 
Again this can be observed both numerically and in the lower bound, even though the inputs only differ in their directionality, not their size. 
Interestingly, both the lower bound and the numerical results also indicate that Exp.~3 poses the hardest identification problem out of the three.
Numerically it can be seen that \ac{MLE} consistently outperforms \ac{OLS} for our setup, reinforcing the interest in its statistical analysis. 
\change{Note that the system considered is strongly damped and hence based on the discussions in Section~\ref{sec:comp} the number of samples only has a weak influence on the upper bound.}
Reducing the conservatism in Theorem\,\ref{th:upperBound} especially in this strictly stable regime remains an important open problem.

\section{Conclusion}
In this paper we provided upper and lower bounds for the sample complexity of identifying an LTI system from a finite set of candidates in absence of stability assumptions. 
These are the first finite sample guarantees for this setting that, albeit relevant, was not studied before.
Future work includes reducing the conservatism in the upper bound, \change{considering additional noise classes} and relaxing the assumption that $\theta_* \notin \mathcal{S}$. 
\change{Finally, the fact that no stability is required here (Remark~\ref{rem:stability}), suggests to better understand whether this depends on the choice of the estimator made here or the finite hypothesis class.}

\bibliographystyle{IEEEtran}
\bibliography{references.bib}

\appendix

\subsection{Proof of Theorem~\ref{th:lowerBound}}\label{app:proofLower}
    Define the data observed up to time $t$ as $\mathcal{D}_t \coloneqq \{x_1, u_1, \dots x_t, u_t\}$ and the probability of the observing $\mathcal{D}_t$ under system $\theta$ as $\Prob_\theta(\mathcal{D}_t)$. Then, we define the log-likelihood ratio of the first $t$ observations under $\theta_*$ and some $\theta_i \in \SSet\setminus\{\theta_*\}$ as 
    $L_t = \log\left(\frac{\Prob_{\theta_*}(\mathcal{D}_t)}{\Prob_{\theta_i}(\mathcal{D}_t)}\right)$.
  Following the change of measurement argument in \cite{jedra2019sample}, we use the generalized data processing inequality \cite[Lemma 1]{Garivier2019} to obtain 
    \begin{align*}
        \Expect_{\theta_*}\left[L_t\right] &= \KL{\Prob_{\theta_*}(\mathcal{D}_t)}{\Prob_{\theta_i}(\mathcal{D}_t)} \\&\ge \sup_{\E \in \mathcal{F}_t} \kl{\Prob_{\theta_*}(\E)}{\Prob_{\theta_i}(\E)}, 
    \end{align*}
    where $\kl{x}{y}$ is the KL-divergence of two Bernoulli distributions of means $x$ and $y$, respectively.
    Since we analyze $\delta$-stable algorithms we define the event $\E \coloneqq \{\hat \theta_t = \theta_*\}$ s.t. consequently $\Prob_{\theta_*}(\E) \ge 1-\delta$ and $\Prob_{\theta_i}(\E) \le \delta$
    and hence 
    \begin{equation*}
        \kl{\Prob_{\theta_*}(\E)}{\Prob_{\theta_i}(\E)} \ge (2\delta -1)\log\left(\frac{1-\delta }{\delta}\right) \ge \log(\nicefrac{1}{2.4\delta}).
    \end{equation*}
    Further, we follow \cite[Section IV.A]{jedra2019sample} to obtain 
    \begin{align*}
        \Expect_{\theta_*}&\left[L_t\right] = \frac12 \Expect_{\theta_*} \left[\sum_{s=0}^{t-1} [\star] [\star] \Sigma_w^{-1} \begin{bmatrix}
            \Delta A_i & \Delta B_i 
        \end{bmatrix} \begin{bmatrix}
            x_s \\ u_s 
        \end{bmatrix} \right] \\
        &= \frac12 \Tr\left([\star]\Sigma_w^{-1} \begin{bmatrix}
            \Delta A_i & \Delta B_i 
        \end{bmatrix} \sum_{s=0}^{t-1} \Expect_{\theta_*}\left[\begin{bmatrix}
            x_s \\ u_s 
        \end{bmatrix} [\star]\right]\right),
    \end{align*}
    where in the last step we used the fact $\Expect\left[X^\top A X\right] = \Tr(A \Expect\left[XX^\top\right])$. 
    Note that up to this point, we have not used that $u \simiid \N(0, \Sigma_u)$.
    By this assumption, we can write 
    $
        \Expect_{\theta_*}\left[[\star]\begin{bmatrix} 
            x_s & u_s 
        \end{bmatrix}\right] = \diag(\Sigma_{x_t}, \Sigma_u),
    $
    where $\Sigma_{x_t}$ is the $t$-step controllability Gramian defined in \eqref{eq:Distributions}. Using $\Tr(ABC) = \Tr(BCA)$ with  $A =\begin{bmatrix}
            \Delta A_i & \Delta B_i 
        \end{bmatrix}^\top \Sigma_w^{-\frac12}$, 
        $B = A^\top $ and   $C= \sum_{s=0}^{t-1} \diag(\Sigma_{x_t}, \Sigma_u)$
    we finally obtain
    \begin{align*}
        \Expect_{\theta_*}\left[L_t\right] = \Tr\left([\star] \left( \sum_{s=0}^{t-1} [\star] \Sigma_u \Delta B_i^\top + [\star] \Sigma_{x_s} \Delta A_i^\top \right) \Sigma_w^{-\frac12}\right)
    \end{align*} 
    and hence for any $\delta$-stable algorithm 
    \begin{align*}
        \Tr\Bigg(\Sigma_w^{-\frac12} \Bigg(\sum_{s=0}^{\bar T-1}  [\star] \Sigma_u \Delta B_i^\top + [\star] &\Sigma_{x_s} \Delta A_i^\top \Bigg) \Sigma_w^{-\frac12}\Bigg) \\ &\ge 2 \log(\frac{1}{2.4 \delta}).
    \end{align*}
    Finally, we similarly as in the proof of Theorem\,\ref{th:lowerBound} obtain  
    \begin{alignat*}{1}
        \Tr\Bigg(\Sigma_w^{-\frac12} \Bigg(\sum_{s=0}^{\bar{T}-1}  \Delta B_i \Sigma_u \Delta {B_i}^\top &+ \Delta A_i \Sigma_{x_s} \Delta A_i^\top \Bigg) \Sigma_w^{-\frac12}\Bigg) \\
    = \bar T \Bigg\Vert \Sigma_w^{-\frac12} \Delta B_i \Sigma_u^{\frac12}\Bigg\Vert _\mathrm{F}^2  + \sum_{s=0}^{\bar T -1} &\sum_{k=0}^{s-1} \left\Vert \Sigma_w^{-\frac12} \Delta A_i A^k B \Sigma_u^{\frac12}\right\Vert _\mathrm{F}^2  \\ &\quad + \left\Vert \Sigma_w^{-\frac12} \Delta A_i A^k \Sigma_w^{\frac12}\right\Vert _\mathrm{F}^2.
    \end{alignat*}

\subsection{Proof of Proposition\,\ref{prop:BMSB}}\label{sec:ProofBMSB}
Recall the definition of the random variable $z_t^i$ \eqref{eq:defz} and its distribution~\eqref{eq:distZ}.
For some $v \in \Sbb^{n_x-1}$ we have $\langle v, z^i_{s+t} \rangle \vert \mathcal{F}_s \sim \mathcal{N}\left(\langle v, \Delta A_i A^t x_s \rangle, v^\top \Sigma_{z_t}^iv\right)$.
Now consider some $k' \le t$ s.t.
\begin{align}
    \Prob\Bigg[\vert \langle v, z^i_{s+t} \rangle\vert &\ge   \sqrt{v^\top \Sigma_{z_{k'}}^iv} \vert \mathcal{F}_s \Bigg] \notag \\ &= \Prob\left[\vert \langle v, z_{t}^i \rangle\vert \ge  \sqrt{v^\top \Sigma_{z_{k'}}^iv} \right] \notag \\ &\ge \Prob\left[\vert \langle v, z^i_{t} \rangle\vert \ge  \sqrt{v^\top \Sigma_{z_t}^iv}\right], \label{eq:ProbIneq}
\end{align}
where the first step follows since the distributions are equal and the inequality follows from $\Sigma_{z_t}^i \succeq \Sigma_{z_{k'}}^i$ for $k' \le t$.
Defining $\zeta_t^i = \langle v, z_t^i-\Delta A_i A ^t x_0\rangle \simiid \mathcal{N}(0,v^\top \Sigma_{z_t}^iv)$ yields
\begin{align*}
    \Prob\Bigg[\vert \langle v, z_t^i \rangle \vert &\ge  \sqrt{v^\top \Sigma_{z_t}^i v} \Bigg]  \\ & = \Prob\left[\vert \zeta_t^i + \langle v, \Delta A_i A^t x_0 \rangle \vert \ge  \sqrt{v^\top \Sigma_{z_t}^i v} \right] \\
                                                                                                                            & \ge \Prob\left[\vert \zeta_t^i \vert \ge   \sqrt{v^\top \Sigma_{z_t}^i v} \right] \ge \frac{3}{10},
\end{align*}
where the last inequality follows from the fact that for any $ \xi \sim \mathcal{N}(0,\sigma^2)$ we have $\Prob[\vert \xi \vert \ge \sigma] \ge \nicefrac{3}{10}$.
It follows that
\begin{align*}
    \frac 1k &\sum_{t=1}^{k}  \Prob\Big[\vert \langle v, z^i_{s+t} \rangle\vert \ge \sqrt{v^\top \Sigma_{z_{k'}}^iv} \vert \mathcal{F}_s\Big] \\ &=  \frac 1k \sum_{t=1}^{k}\Prob\left[\vert \langle v, z_t^i \rangle \vert \ge   \sqrt{v^\top \Sigma_{z_{k'}}^i v} \right]   \\
                                                                                                                                                                             & \ge  \frac 1k \sum_{t=k'}^{k}  \Prob\left[\vert \langle v, z_t^i \rangle \vert \ge   \sqrt{w^\top \Sigma_{z_t}^i w} \right] \ge \frac{3}{10} \frac{k-k'+1}{k} .
\end{align*}
The result then follows by choosing $k' = \frac{1}{2}k$.
\subsection{Using the BMSB condition to show anti-concentration}\label{app:SimchoTight}
If a sequence $(z_t)_{t\ge 0}$ satisfies the BMSB condition, anti-concentration can be shown. In the following, we provide a milder version of \cite[Proposition 2.5]{simchowitz2018learning}, in which the probability bound scales with $p$ instead of $p^2$.  
\begin{corollary}\label{co:SimchoTight}
    Suppose that $(z_1, \dots, z_T) \in \R^T$ satisfies the $(k, \nu, p)$-BMSB condition.
    Then 
    \begin{equation*}
        \Prob\left[\sum_{t=1}^T z_t^2 \le \frac{\nu^2 p^2}{8} k \lfloor\nicefrac{T}{k}\rfloor\right] \le \exp\left(-\frac{\lfloor \nicefrac{T}{k}\rfloor p}{16}\right).
    \end{equation*}
\end{corollary}
\begin{proof}
    The proof of Corollary~\ref{co:SimchoTight} builds on the original work, thus we only provide a sketch here. 
    Set $S =\lfloor\nicefrac{T}{k}\rfloor$ 
    and  
    \begin{equation*}
        B_j = \mathbb{I} \left[\sum_{t=1}^k z_{jk+t}^2 \ge \frac{\nu^2 p k}{2}\right] \quad \forall j\in [0, S-1].    
    \end{equation*}
    Using the same arguments as in the original proof we obtain
    \begin{equation}
        \Prob\Bigg[\sum_{t=1}^T z_t^2 \le \frac{\nu^2 p^2}{8} k S\Bigg]  \le \inf_{\lambda\le 0} e^{-\lambda S\frac{p}{4}} \Expect\left[e^{\lambda \sum_{j=0}^{S-1} B_j}\right]. \label{eq:ProofImproved1}
    \end{equation}
    Inserting the optimizer $\lambda_* = \log(\frac{1-\frac{p}{2}}{2-\frac{p}{2}})$ of \eqref{eq:ProofImproved1}, yields
    \begin{align*}
        \Prob&\Bigg[\sum_{t=1}^T z_t^2 \le \frac{\nu^2 p^2}{8} k S\Bigg] \le \left(\left(\frac{2-p}{4-p}\right)^{-\nicefrac{p}{4}} \left(1-\frac{p}{4-p}\right)\right)^{S}\\
&\le \left(2^{\nicefrac{p}{4}} \left(1-\frac{p}{4-p}\right)\right)^{S}
 \le e^{-\frac{S}{4}p \left(\frac{p}{4} + 1- \log(2)\right)} \le e^{-\frac{S}{16}p},
    \end{align*}
where the last step follows from $p>0$ and $1- \log(2) >\nicefrac{1}{4}$.
\end{proof}
\end{document}